\documentclass[11pt]{article}

\usepackage{amsfonts,amssymb,amsmath,amsthm,bm}
\usepackage{latexsym,url,color,enumerate}
\usepackage{graphicx}
\usepackage{amsfonts,amsmath,amssymb,amsthm,caption,algorithm,algorithmic}
\usepackage{lipsum}%footnote
\usepackage{amsbsy}
\usepackage{mathrsfs}
\usepackage{graphicx,multirow,bm}
\usepackage{color}
\usepackage{enumerate}
\usepackage{diagbox}
\usepackage{amssymb}
\usepackage{psfrag}
\usepackage{bbm}
\usepackage[usenames,dvipsnames]{pstricks}
\usepackage{epsfig}
\usepackage{pst-grad} % For gradients
\usepackage{pst-plot} % For axes
\usepackage{upgreek}

\allowdisplaybreaks[4]

\newtheorem{theorem}{Theorem}%[section]

\newtheorem{lemma}[theorem]{Lemma}

\newtheorem{problem}[theorem]{Problem}
\newtheorem{proposition}[theorem]{Proposition}

\newtheorem{conjecture}[theorem]{Conjecture}
\theoremstyle{definition}
\newtheorem{definition}[theorem]{Definition}
\newtheorem{remark}[theorem]{Remark}
\usepackage{hyperref}

\begin{document}
\title{On Compressed Sensing Matrices Breaking\\ the Square-Root Bottleneck}
\author{Shohei Satake
\thanks{Faculty of Advanced Science and Technology, 
Kumamoto University, Kumamoto 860-8555, Japan 
\{email: shohei-satake@kumamoto-u.ac.jp\}. 
This work has been supported by Grant-in-Aid for JSPS Fellows 20J00469 of the Japan Society for the Promotion of Science.
}
\ \ and \ 
Yujie Gu
\thanks{Department of Electrical Engineering - Systems, Tel Aviv University, Tel Aviv, Israel 
and 
Faculty of Information Science and Electrical Engineering, Kyushu University, Fukuoka, Japan
\{email: gu@inf.kyushu-u.ac.jp\}.
}
}

\date{}
\maketitle

\begin{abstract}
Compressed sensing is a celebrated framework in signal processing and has many practical applications. 
One of challenging problems in compressed sensing is to construct deterministic matrices having restricted isometry property (RIP).
So far, there are only a few publications providing deterministic RIP matrices  beating the square-root bottleneck on the sparsity level. 
In this paper, we investigate RIP of certain matrices defined by higher power residues modulo primes.
Moreover, we prove that the widely-believed generalized Paley graph conjecture implies that these matrices have RIP breaking the square-root bottleneck. 
\end{abstract}

%\begin{IEEEkeywords}
%
%\end{IEEEkeywords}

\section{Introduction}
\label{sect-intro}
Matrices with {\it restricted isometry property} ({\it RIP}) have important applications to compressed processing. According to~\cite{Ca2008}, by means of RIP matrices, it is possible to measure and recover sparse signals using significantly fewer measurements than the dimension of the signals. 

\begin{definition}[Restricted isometry property, RIP]
Let $\Phi$ be a complex $M \times N$ matrix. Suppose that $K \leq M \leq N$ and $0 \leq \delta<1$. 
Then $\Phi$ is said to have the {\it $(K, \delta)$-restricted isometry property} ({\it RIP}) if 
\begin{equation}
    (1-\delta)||\mathbf{x}||^2 \leq ||\Phi \mathbf{x}||^2 \leq  (1+\delta)||\mathbf{x}||^2
\end{equation}
for every $N$-dimensional complex vector $\mathbf{x}$ with at most $K$ non-zero entries. Here $||\cdot||$ denotes the $\ell_2$ norm.
\end{definition}
According to Cand\`{e}s~\cite{Ca2008}, for applications to signal processing, it suffices to investigate the $(K, \delta)$-RIP matrix for some $\delta<\sqrt{2}-1$. 
In addition, the {\it sparsity} $K$ is expected to be as large as possible. 

On the other hand, it is known (\cite{BDMS2013}) that the problem checking whether a given matrix has RIP is NP-hard.
Thus many publications have attempted to look for deterministic constructions of matrices having RIP; see e.g. \cite{RMKADD2020}.  

Throughout this paper, we assume that all matrices have column vectors with unit $\ell_2$-norm.
Most of known constructions for RIP matrices are established via the {\it coherence} $\mu(\Phi)$ of an $M \times N$ matrix $\Phi$ with column vectors $\psi_1, \ldots, \psi_N$, where 
\begin{equation}
    \mu(\Phi):=\max_{1 \leq j\neq k \leq N}|\langle \psi_j, \psi_k \rangle|,
\end{equation}
and $\langle \cdot, \cdot \rangle$ denotes the standard inner product in the Hilbert space $\mathbb{C}^M$.
It can be proved (e.g. \cite{BDFKK2011}) that if $\mu(\Phi)=\mu$, then $\Phi$ has the $(K, (K-1)\mu)$-RIP, which implies the $(K, \delta)$-RIP with only $K=O(\sqrt{M})$, following from the well-known {\it Welch bound} (\ref{eq-Welch}) in \cite{W1974}. 
\begin{equation}
\label{eq-Welch}
    \mu(\Phi) \geq \sqrt{\frac{N-M}{M(N-1)}}.
\end{equation}
This barrier on the magnitude of the order of $K$ is popularly dubbed as the {\it square-root bottleneck} or {\it quadratic bottleneck}.
Accordingly, the following problem arises.

\begin{problem}[\cite{BDFKK2011}]
\label{prob-beatsqrt}
Construct an $M \times N$ matrix $\Phi$ having the $(K, \delta)$-RIP with $K=\Omega(M^{\gamma})$ for some $\gamma>1/2$ and $\delta<\sqrt{2}-1$.
\end{problem}
To our best knowledge, the first (unconditional) solution to this problem was given by Bourgain, Dilworth, Ford, Konyagin and Kutzarova~\cite{BDFKK2011}, later improved by Mixon~\cite{M2015}; see Table 1.%~\ref{table-1}.

On the other hand, in \cite{BFMW2013}, Bandeira, Fickus, Mixon and Wong conjectured that the {\it Paley matrix}, a $(p+1)/2 \times (p+1)$ matrix defined by quadratic residues modulo an odd prime $p$, satisfies the $(K, \delta)$-RIP with $K\geq C_1\cdot p/\log^{C_2} p$ and some $\delta<\sqrt{2}-1$, where $C_{1}, C_{2}>0$ are universal constants.  
Under a number-theoretic conjecture in \cite{C1994}, 
Bandeira, Mixon and Moreira~\cite{BMM2017} proved that when $p \equiv 1 \pmod{4}$, the Paley matrix has the $(K, o(1))$-RIP with $K=\Omega(p^{\gamma})$ for some $\gamma>1/2$, which provides a conditional solution to Problem~\ref{prob-beatsqrt}; 
recently, assuming that the Paley graph conjecture (see Remark~\ref{rem-originalpgc}) holds, Satake~\cite{S2020} extended the result in~\cite{BMM2017} for general odd primes $p$, and also gave some implications to a Ramsey-theoretic problem proposed by Erd\H{o}s and Moser~\cite{EM1964}.
Also, under another type of number-theoretic conjecture, 
Arian and Yilmaz~\cite{AY2019} proved that for a sufficiently large prime $p \equiv 3 \pmod{4}$, a $(p+1)/2 \times p$ matrix obtained by deleting the last column from the Paley matrix (see Remark~\ref{rmk-Paley}) has the $(K, \delta)$-RIP for any $K<p^{5/7}/2$ and $\delta<1/\sqrt{2}$.
These results are summarized in Table 1.%~\ref{table-1}.%\Satake{Please check the table!}

\begin{table}[htbp]
\label{table-1}
\tabcolsep = 3pt
\begin{center}
\scalebox{0.95}{
  \begin{tabular}{|c|c|c|c|c|} \hline
    Ref. & $M$ & $N$ & $K$ & Comments\\ \hline 
    \cite{BDFKK2011}, \cite{M2013} & large $M$ & $N\leq M^{1+\gamma_1}$ & $\lfloor M^{1/2+\gamma'_1} \rfloor$ & any $\gamma'_1<\gamma_1$ and $\gamma_1 \fallingdotseq 5.5169 \times 10^{-28}$ \\ \hline
    \cite{M2015} & $p$ & $\Omega(p^{1+\eta})$ & $\Omega(p^{1/2+\gamma'_2})$ & some $\eta>0$, any $\gamma'_2<\gamma_2$ and $\gamma_2 \fallingdotseq 4.4466 \times 10^{-24}$ \\ \hline \hline 
    \cite{BMM2017} & $\frac{p+1}{2}$ & $p+1$ & $\Omega(p^{\gamma})$ & \begin{tabular}{c}
    some $\frac{1}{2}<\gamma<1$, 
    $p \equiv 1 \pmod{4}$
    \end{tabular}
    \\ \hline
    \cite{S2020} & $\frac{p+1}{2}$ & $p+1$ & $\Omega(p^{\gamma})$ & \begin{tabular}{c}
    some $\frac{1}{2}<\gamma<1$, 
    $p \equiv 3 \pmod{4}$
    \end{tabular}
    \\ \hline
    \cite{AY2019} & $\frac{p+1}{2}$ & $p$ & $< \frac{p^{5/7}}{2}$ & 
    \begin{tabular}{c}
    $p \equiv 3 \pmod{4}$
    \end{tabular}
    \\ \hline
    \begin{tabular}{c}
    This paper
    \end{tabular} & 
    \begin{tabular}{c}
    $\frac{p+k-1}{k}$
    \end{tabular}
    & $p$ & $\Omega(M^{\gamma})$ & 
    \begin{tabular}{c}
    some $\frac{1}{2}<\gamma<1$,\\
    $k|(p-1)$ s.t. $p^{\varepsilon_1}<k\leq p^{\varepsilon_2}$,\\
    any $0\leq \varepsilon_1<\varepsilon_2<\varepsilon_0$, \\
    some small $\varepsilon_0>0$\\
    \end{tabular}
    \\ \hline
  \end{tabular}
  }
  \end{center}
  \caption{Main result and known solutions to Problem~\ref{prob-beatsqrt}. Here $p$ denotes an odd prime number.
  }
\end{table}

In this paper, we aim to investigate RIP of certain matrices defined by higher power residues modulo primes.
In particular, under the widely-believed {\it generalized Paley graph conjecture} formulated in Section~\ref{sect-prelim}, we prove that these matrices are new solutions to Problem~\ref{prob-beatsqrt}, which forms our main theorem as described below.

\begin{theorem}
\label{thm-main}
Suppose that the generalized Paley graph conjecture holds. 
Let $\varepsilon_0>0$ be a small real number and $\varepsilon_1, \varepsilon_2$ real numbers with
$0\leq \varepsilon_1<\varepsilon_2<\varepsilon_0$.
Then for a sufficiently large prime $p$ such that $p-1$ contains a factor $k$ with $p^{\varepsilon_1}<k\leq p^{\varepsilon_2}$,
there are $M \times p$ matrices with $(\Omega(M^{\gamma}), o(1))$-RIP for some $\gamma>1/2$, where $M=(p+k-1)/k$.
\end{theorem}

Theorem~\ref{thm-main} provides a new conditional answer to Problem~\ref{prob-beatsqrt}. 
As will be shown later, these matrices substantially contain the Paley matrix as a particular case, and realize the {\it compression ratio} $N/M$ (e.g. \cite{LLKR2015}) significantly better than that from the Paley matrix in general; see Remarks~\ref{rmk-Paley} and \ref{rmk-compratio}, respectively.

The remainder of this paper is organized as follows.
Section~\ref{sect-prelim} introduces some fundamental terminologies and results of finite fields and character sums, as well as some key notions related to RIP and number-theoretic results on factors of shifted primes.
Section~\ref{sect-const} defines the matrix which we investigate in this paper, and then Section~\ref{sect-mainproof} proves Theorem~\ref{thm-main}. 

\section{Preliminaries}
\label{sect-prelim}

\subsection{Finite fields and characters}
Throughout this paper, let $p$ denote a prime number.
Let $\mathbb{F}_p$ be a finite field with $p$ elements which can be identified to the residue ring $\mathbb{Z}/p\mathbb{Z}$.
It is well known that the {\it multiplicative group} of $\mathbb{F}_p$, denoted by $\mathbb{F}_p^*$, is a cyclic group of order $p-1$, consisting of all non-zero elements of $\mathbb{F}_p$.

The {\it canonical additive character} $\psi$ of $\mathbb{F}_p$ is a map from $\mathbb{F}_p$ to the unit circle in $\mathbb{C}$ such that $\psi(x):=\exp(\frac{2 \pi \sqrt{-1}}{p} \cdot x)$ for all $x \in \mathbb{F}_p$.
Notice that for every pair of $x, y \in \mathbb{F}_p$, we have $\psi(x+y)=\psi(x)\psi(y)$.
A {\it multiplicative character} $\chi$ of $\mathbb{F}_p$ is a map from $\mathbb{F}_p^*$ to the unit circle in $\mathbb{C}$ such that $\chi(xy)=\chi(x)\chi(y)$ for every pair of $x, y \in \mathbb{F}_p^*$. 
We also adopt the convention that $\chi(0):=0$.
Note that if $g$ is a generator of $\mathbb{F}_p^*$, then each multiplicative character $\chi$ is a map such that $\chi(x):=\exp(\frac{2 \pi \sqrt{-1}}{p-1}st)$ for all $x=g^t \in \mathbb{F}_p^*$ and some $0\leq s \leq p-2$.
The multiplicative character of $\mathbb{F}_p$ for $s=0$ is said to be {\it trivial}. 
The {\it order} of a multiplicative character $\chi$ is the minimum positive integer $k$ such that $\big(\chi(x)\big)^k=1$ for all $x \in \mathbb{F}_p^*$.
Notice that the order of the trivial multiplicative character of $\mathbb{F}_p$ is $1$.
For each $2 \leq k \leq p-2$ with $k|(p-1)$,
we can define a non-trivial multiplicative character $\chi_k$ of $\mathbb{F}_p$ of order $k$, that is,  $\chi_k(x):=\exp(\frac{2 \pi \sqrt{-1}}{k}t)$ for all $x=g^t \in \mathbb{F}_p^*$.
For each $1\leq h \leq k-1$, define the multiplicative character $\chi_k^{-h}$ of $\mathbb{F}_p$ as $$\chi_k^{-h}(x):=\exp\bigg(\frac{2 \pi \sqrt{-1}}{k}(-ht)\bigg)$$ 
for all $x=g^t \in \mathbb{F}_p^*$.
Notice that $\chi_k^{-h}$ is non-trivial for every $1\leq h \leq k-1$.
%\Satake{Define $\chi_k^{-h}$!!!!}

\subsection{Gauss sums}
This subsection provides two types of {\it Gauss sums}; for details, see e.g. \cite{BEW1998}.
\begin{definition}[$k$-th power Gauss sum]
Let $2 \leq k \leq p-2$.
Then for each $a \in \mathbb{F}_p$, the {\it $k$-th power Gauss sum} $\mathcal{G}_k(a)$ is defined as 
\begin{equation}
    \mathcal{G}_k(a):=\sum_{x \in \mathbb{F}_p}\psi(ax^k).
\end{equation}
\end{definition}

\begin{definition}[Gauss sum]
Let $a \in \mathbb{F}_p$ and $\chi$ a multiplicative character of $\mathbb{F}_p$.
Then define the {\it Gauss sum} $G(a, \chi)$ as
\begin{equation}
  G(a, \chi):=\sum_{x \in \mathbb{F}_p}\chi(x)\psi(ax).
\end{equation}
For simplicity, let $G(\chi):=G(1, \chi)$.
\end{definition}

The following two lemmas are useful in this paper. 

\begin{lemma} (\cite[Theorem 1.1.3 and (1.1.4)]{BEW1998})
\label{lem-Gauss}
For each $a \in \mathbb{F}_p^*$, it holds that
\begin{equation}
\mathcal{G}_k(a)=\sum_{h=1}^{k-1}G(a, \chi_k^h)
=\sum_{h=1}^{k-1}\chi_k^{-h}(a)G(\chi_k^h).
\end{equation}
\end{lemma}

%{\color{red} Fix $\sqrt{-1}$!}
\begin{lemma}(\cite[Theorem 1.1.4]{BEW1998})
\label{lem-Gauss2}
For each $1 \leq h \leq k-1$, it holds that
\begin{equation}
    |G(\chi_k^h)|=\sqrt{p}.
\end{equation}
\end{lemma}
%Here notice that $\chi_k^h$ is non-trivial for each $1 \leq h \leq k-1$.

\subsection{Generalized Paley graph conjecture}
In this paper, we will make use of the following well-known {\it generalized Paley graph conjecture}; see e.g. \cite{C2008}, \cite{CZ2016}, \cite{GM2020}, \cite{K2008}, \cite{S2014}, \cite{Z1991} and references therein.

%focus on Conjecture~\ref{conj-pgc} which is known as the {\it Paley graph conjecture} (see e.g. \cite{C2008}, \cite{GM2020}). 
% For the history and implications of this conjecture, see e.g. \cite{C2008}, \cite{GM2020}.

\begin{conjecture}[Generalized Paley graph conjecture]
\label{conj-pgc}
Let $p$ be an odd prime and $\chi$ a non-trivial multiplicative character of $\mathbb{F}_p$.
For $0<\alpha \leq 1$ and $\beta>0$, we say that the property $\mathcal{P}(\alpha, \beta)$ holds if for every pair of $S, T \subset \mathbb{F}_p$ with $|S|, |T|>p^{\alpha}$,
\begin{equation}
\label{eq-pgc}
   \Bigl |\sum_{s \in S, t \in T}\chi(s-t) \Bigr|\leq p^{-\beta}|S||T|.
\end{equation}
Then for each $0<\alpha \leq 1$, there exist $p(\alpha)>0$ and $\beta=\beta(\alpha)>0$ such that $\mathcal{P}(\alpha, \beta)$ holds for any prime $p>p(\alpha)$. 
\end{conjecture}

%Let us introduce a connection between Conjecture~\ref{conj-pgc} and Paley graphs.
%For a prime $p \equiv 1 \pmod{4}$, the {\it Paley graph} $G_p$ with $p$ vertices is defined as an undirected graph with vertex set $\mathbb{F}_p$ and two distinct vertices $x$ and $y$ are adjacent if and only if $\chi_2(x-y)=1$.
%The clique number (i.e. the size of largest cliques), denoted by $\omega(G_p)$, of $G_p$ has been extensively studied in, for example, Ramsey theory and additive combinatorics. 
%It is not so difficult to see that Conjecture~\ref{conj-pgc} implies that $\omega(G_p)\leq p^{\varepsilon}$ for any $\varepsilon>0$ and sufficiently large $p$.

\begin{remark}
\label{rem-originalpgc}
The {\it Paley graph conjecture} is Conjecture~\ref{conj-pgc} for the case that $\chi=\chi_2$, where $\chi_2$ is a non-trivial multiplicative character of $\mathbb{F}_p$ of order $2$.
According to the observations by Lenstra (see~\cite{Z1991}), it may be reasonable to believe that Conjecture~\ref{conj-pgc} would be true for any non-trivial $\chi$ if the Paley graph conjecture holds.
Indeed, it is well-known (e.g. \cite{K2008}) that Conjecture~\ref{conj-pgc} is true for {\it any} non-trivial multiplicative character $\chi$ of $\mathbb{F}_p$ and any $S, T \subset \mathbb{F}_p$ with $|S|>p^{1/2+\alpha}$ and $|T|>p^{\alpha}$, where $0<\alpha \leq 1/2$.
In~\cite{C2008}, Chang made a significant progress towards Conjecture~\ref{conj-pgc}, confirming the conjecture for {\it any} non-trivial $\chi$ and any $S, T \subset \mathbb{F}_p$ such that $|S|>p^{4/9+\alpha}$, $|T|>p^{4/9+\alpha}$ with $|T+T|<K|T|$ for some $K>0$ and arbitrary $\alpha>0$, where $T+T:=\{t_1+t_2 \mid t_1, t_2 \in T\}$.
For further and related results, see e.g. \cite{H2017}, \cite{S2013}, \cite{V2018}, \cite{VS2017}.
\end{remark}

\subsection{A theorem on factors of shifted primes}
To prove Theorem~\ref{thm-main}, it is necessary to certify the existence of infinitely many primes $p$ such that the {\it shifted prime} $p-1$ admits the prescribed factors.
In this paper, we shall use the following theorem due to Ford~\cite[Theorem 7]{F2008}.

\begin{theorem}[\cite{F2008}]
\label{thm-F2008}
For any $0\leq \varepsilon_1<\varepsilon_2\leq 1$ and any sufficiently large $x>0$, there exists a constant $C_{\varepsilon_1, \varepsilon_2}>0$ depending only on $\varepsilon_1$ and $\varepsilon_2$ such that there exist at least $C_{\varepsilon_1, \varepsilon_2} \cdot x/\log x$ primes $p\leq x$ so that $p-1$ contains at least one factor $k$ with $x^{\varepsilon_1}<k \leq x^{\varepsilon_2}$.
In particular, there exist infinitely many primes $p$ such that $p-1$ contains at least one factor $k$ with $p^{\varepsilon_1}<k \leq p^{\varepsilon_2}$.
\end{theorem}

\subsection{Flat restricted isometry property}
The following notion, {\it flat restricted isometry property} ({\it flat RIP}), has been  employed to verify the RIP of matrices in~\cite{BFMW2013} and~\cite{BDFKK2011}.

\begin{definition}[Flat RIP, \cite{BFMW2013,BDFKK2011}]
Let $\Phi$ be an $M \times N$ matrix with columns $\psi_1, \ldots, \psi_N$. Suppose that $K \leq M \leq N$ and $\theta>0$. 
Then $\Phi$ is said to have the {\it $(K, \theta)$-flat restricted isometry property} ({\it flat RIP}) if 
\begin{align}
   \Bigl |\Bigl\langle \sum_{i \in I}\psi_i, \sum_{j\in J}\psi_j \Bigr\rangle \Bigr|\leq \theta \sqrt{|I||J|}
\end{align}
for every pair of disjoint subsets $I, J \subset \{1, 2, \ldots, N\}$ with $|I|, |J| \leq K$.
\end{definition}

\begin{proposition}[\cite{BFMW2013,BDFKK2011}]
\label{prop-flatRIP}
Suppose that each column of $\Phi$ has unit $\ell_2$ norm. 
Then $\Phi$ has the $(K, 150\: \theta \log K)$-RIP provided that $\Phi$ has the $(K, \theta)$-flat RIP.
\end{proposition}

\section{Construction of matrices}
\label{sect-const}
In this section, we define a type of matrices based on higher power residues modulo primes and later investigate their RIP in Section~\ref{sect-mainproof}. 

\begin{definition}
\label{def-Phi}
Let $p$ be an odd prime and $k|(p-1)$.
%Let $\mathbb{F}_p$ denote the finite field with $p$ elements. 
Let $R_p^{(k)}:=\{x^{k} \mid x \in \mathbb{F}_p^*\}$ denote the set of all non-zero $k$-th powers of $\mathbb{F}_p$; notice that $|R_p^{(k)}|=(p-1)/k$.
Suppose that elements of $\mathbb{F}_p$ and $R_p^{(k)}$ are labelled as $\mathbb{F}_p=\{0=a_1, a_2, \ldots, a_p\}$ and $R_p^{(k)}=\{b_1, b_2, \ldots, b_{(p-1)/k}\}$, respectively.
Recall that $\psi$ denotes the canonical additive character of $\mathbb{F}_p$.

Then the matrix $\Phi_p^{(k)}$ is defined as an $M \times N$ complex matrix with $M=\frac{p+k-1}{k}$ and $N=p$ of the following form.
\begin{align*}
\begin{split}
  \Phi_p^{(k)} :=
  \left[
    \begin{array}{cccc}
      \frac{1}{\sqrt{p}} & \frac{1}{\sqrt{p}} & \ldots & \frac{1}{\sqrt{p}}  \\ \\
      \sqrt{\frac{k}{p}} & \sqrt{\frac{k}{p}}\psi(b_1a_{2}) & \ldots & \sqrt{\frac{k}{p}}\psi(b_1a_{p})  \\ \\
      \sqrt{\frac{k}{p}} & \sqrt{\frac{k}{p}}\psi(b_2a_{2}) & \ldots & \sqrt{\frac{k}{p}}\psi(b_2a_{p}) \\ \\
      \vdots & \vdots & \ddots & \vdots \\ \\
      \sqrt{\frac{k}{p}} & \sqrt{\frac{k}{p}}\psi \bigl(b_{\frac{p-1}{k}}a_{2} \bigr) & \ldots & \sqrt{\frac{k}{p}}\psi \bigl(b_{\frac{p-1}{k}}a_{p} \bigr) 
    \end{array}
  \right]
  \end{split}
\end{align*}
\end{definition}
Note that each column $\phi_i$ of $\Phi_p^{(k)}$ has $\ell_2$-norm $1$ since for each $1 \leq i \leq p$,
\begin{align*}
    \begin{split}
    ||\phi_i||^2
    &=\langle \phi_i, \phi_i \rangle\\
    &=\frac{1}{p}+\frac{k}{p}\sum_{l=1}^{\frac{p-1}{k}}\psi\bigl((a_i-a_i)b_l \bigr)\\
    &=\frac{1}{p}+\frac{k}{p} \cdot \frac{p-1}{k}
    =1.\\
\end{split}
\end{align*}

\begin{remark}
\label{rmk-Paley}
Let $p$ be an odd prime.
The {\it Paley matrix} in \cite{BMM2017} and \cite{S2020} is a $(p+1)/2 \times (p+1)$ matrix obtained from the matrix $\Phi_p^{(2)}$ joining the vector $[(\sqrt{-1})^r, 0, \ldots, 0]^T$ as the last column, where $r=0$ if $p\equiv 1 \pmod{4}$ and $r=1$ if $p\equiv 3 \pmod{4}$.
%Note that all column vectors of the Paley matrix form an equiangular tight frame~\cite{R2007}.
The matrix investigated in \cite{AY2019} is exactly $\Phi_p^{(2)}$.
\end{remark}

\begin{remark}
\label{rmk-compratio}
In terms of applications to signal processing, it is desirable to construct an $M \times N$ matrix with RIP whose {\it compression ratio} $N/M$ (e.g. \cite{LLKR2015}) is as large as possible. 
%\Satake{Reference?}
By Definition~\ref{def-Phi}, the compression ratio of $\Phi_p^{(k)}$ is $kp/(p+k-1) \approx k$ as long as $k=o(p)$, while the compression ratio of the Paley matrix is only $2$.
Thus when $k\geq 3$, $\Phi_p^{(k)}$ has compression ratio significantly better than that from the Paley matrix in general. 
\end{remark}

%%%%%%%%%%%%%%

\section{Proof of theorems}
\label{sect-mainproof}
This section aims to prove the main result Theorem~\ref{thm-main}. 
To that end, we shall prove the 
following Theorem~\ref{thm-main1}. 

\begin{theorem}
\label{thm-main1}
Assuming that Conjecture~\ref{conj-pgc} is true, and \begin{enumerate}
    \item[$1)$] 
    let $0<\alpha<1/2$ be a real number;
    \item[$2)$] 
    let $\beta_0=\beta_0(\alpha)>0$ be a real number such that $\alpha+2\beta_0<1/2$ and the property $\mathcal{P}(\alpha, \beta_0)$ holds for any non-trivial multiplicative character of $\mathbb{F}_q$ and any prime $q>p(\alpha)$, where $p(\alpha)>0$ is from Conjecture~\ref{conj-pgc};
    \item[$3)$]
    take real numbers $\varepsilon_1$ and $\varepsilon_2$ with $0 \leq \varepsilon_1<\varepsilon_2<\beta_0$;
    
    \item[$4)$]
    let $p>p(\alpha)$ be a prime such that there exists a factor $k$ of $p-1$ with $p^{\varepsilon_1}< k \leq p^{\varepsilon_2}$.
\end{enumerate}
Then 
for any real number $\tau$ with 
\begin{equation}
\label{eq-cond-5}
    \max \Big\{\alpha+\beta_0, \frac{1-\varepsilon_1}{2}-\beta_0 \Big\}<\tau<\frac{1}{2}-\varepsilon_2,
\end{equation}
the $(p+k-1)/k \times p$ matrix $\Phi_p^{(k)}$ in Definition~\ref{def-Phi} has the $(p^{\tau+\beta_0}, O(p^{\tau+\varepsilon_2-1/2+o(1)}))$-RIP.
\end{theorem}

\begin{remark}
\label{rmk-main1}
We remark that for each $\alpha$, $\beta_0$, $\varepsilon_1$ and $\varepsilon_2$ satisfying the conditions 1), 2) and 3), it is possible to take a real number $\tau$ satisfying (\ref{eq-cond-5}).
Indeed, notice that $\alpha+\beta_0<1/2-\varepsilon_2$ holds since 
$\alpha+\beta_0<1/2-\beta_0<1/2-\varepsilon_2$, 
where the first and second inequalities follow from the conditions 2) and 3), respectively.
Also $(1-\varepsilon_1)/2-\beta_0<1/2-\varepsilon_2$ is valid, since $\varepsilon_2-\varepsilon_1/2<\beta_0$ holds for any  $\varepsilon_1$ and $\varepsilon_2$ satisfying the condition 3).
\end{remark}

Before proceeding to prove  Theorem~\ref{thm-main1}, we show that Theorem~\ref{thm-main} can be immediately derived from Theorems~\ref{thm-F2008} and \ref{thm-main1}.

\begin{proof}[Proof of Theorem~\ref{thm-main}] 
According to Theorem~\ref{thm-main1}, 
\begin{itemize}
    \item let $0<\alpha<1/2$ and $\beta_0=\beta_0(\alpha)>0$ be a real number satisfying the condition 2);
    \item take $\varepsilon_0=\varepsilon_0(\alpha)=\beta_0$ and then choose $\varepsilon_1$ and $\varepsilon_2$ so that $0\leq \varepsilon_1<\varepsilon_2<\varepsilon_0$;
    \item take a prime $p>p(\alpha)$ satisfying the condition 4), which is doable by Theorem~\ref{thm-F2008}.
\end{itemize}
Note that based on the condition 4), we have $p^{\varepsilon_1}< k \leq p^{\varepsilon_2}$, implying $M=(p+k-1)/k=O(p^{1-\varepsilon_1})$. Therefore Theorem~\ref{thm-main1} 
shows that the $M\times p$ matrix $\Phi_p^{(k)}$ has the $(K, o(1) )$-RIP, where $$K=p^{\tau+\beta_0}=\Omega(M^{\gamma})$$
and $\gamma$ is a real number such that $\gamma \geq (\tau+\beta_0)\cdot (1-\varepsilon_1)^{-1}>1/2$. 
This proves Theorem~\ref{thm-main}. 
\end{proof}

In order to prove Theorem~\ref{thm-main1}, we need the following two key lemmas. 

\subsection{Two key lemmas}
\begin{lemma}
\label{lem-doubsum}
Let $0<\alpha<1/2$ and $p$ be a prime with $p>p(\alpha)$.
Let $\chi$ be a non-trivial multiplicative character of $\mathbb{F}_p$.
Suppose that there exists $\beta=\beta(\alpha)>0$ such that  
$\alpha+\beta<1/2$ and the property $\mathcal{P}(\alpha, \beta)$ holds.
Let $\tau$ be an arbitrarily fixed real number with $\alpha+\beta<\tau < 1/2$. 
Then it holds that
\begin{equation}
\label{eq-flatpgc}
     \Bigl |\sum_{s \in S, t \in T}\chi(s-t) \Bigr| \leq p^{\tau}\sqrt{|S||T|}
\end{equation}
for every pair of $S, T \subset \mathbb{F}_p$ with $|S|, |T| \leq p^{\tau+\beta}$.
\end{lemma}

\begin{proof}
%Fix $0<\alpha < 1/2$ and $\alpha+\beta<\tau < 1/2$ arbitrarily and take $\beta=\beta(\alpha)>0$ such that the property $\mathcal{P}(\alpha, \beta)$ holds. 
Let $S, T \subset \mathbb{F}_p$ with $|S|, |T| \leq p^{\tau+\beta}$.
The proof is done by considering the following cases.

{\it Case 1.} 
If $|S||T|\leq p^{2\tau}$, then, by the trivial bound of $|\sum_{s \in S, t \in T}\chi(s-t)|$, we have
\begin{align*}
    \Bigl |\sum_{s \in S, t \in T}\chi(s-t) \Bigr|
    \leq |S||T|
    &=\sqrt{|S||T|}\cdot \sqrt{|S||T|} \\
    &\leq p^{\tau}\sqrt{|S||T|}.
\end{align*}

{\it Case 2.} 
Next, suppose that $|S||T|> p^{2\tau}$ and we may assume $|S|>p^{\tau}$ without loss of generality.

{\it Case 2.1.} 
If $|T| \leq p^{\alpha}$, then the following inequalities hold by the assumption that $|S|\leq p^{\tau+\beta}$$\colon$
\begin{align*}
    \Bigl |\sum_{s \in S, t \in T}\chi(s-t) \Bigr|
    \leq |S||T|
    &=\sqrt{|S||T|}\cdot \sqrt{|S||T|}\\ 
    & \leq \sqrt{p^{\alpha}\cdot p^{\tau+\beta}}\sqrt{|S||T|}\\
    &= p^{\frac{\tau+\alpha+\beta}{2}}\sqrt{|S||T|}\\
    &<p^{\tau}\sqrt{|S||T|},
\end{align*}
where the last inequality follows from the assumption that $\alpha+\beta<\tau$.

{\it Case 2.2.} 
If $|T|>p^{\alpha}$, recall that $|S|>p^{\tau}>p^{\alpha+\beta}>p^{\alpha}$, 
by the assumption we have the property $\mathcal{P(\alpha, \beta)}$ holds for $S$ and $T$, that is, 
\begin{align*}
    \Bigl |\sum_{s \in S, t \in T}\chi(s-t) \Bigr|
    &\leq p^{-\beta}|S||T|.
\end{align*}
Together with the assumption that $|S|, |T|\leq p^{\tau+\beta}$, we obtain
\begin{align*}
     \Bigl |\sum_{s \in S, t \in T}\chi(s-t) \Bigr|
    &\leq %p^{-\beta}|S||T|\\
    %&=
    p^{-\beta} \sqrt{|S||T|}\cdot \sqrt{|S||T|} \\
    & \leq p^{-\beta} \cdot \sqrt{p^{2(\tau+\beta)}} \cdot \sqrt{|S||T|}\\
    &=p^{\tau}\sqrt{|S||T|}.
\end{align*}
\end{proof}

\begin{lemma}
\label{lem-inner}
Let $\phi_i$ be the $i$-th column of $\Phi_p^{(k)}$. Then, for each $1 \leq i\neq j \leq p$, 
\begin{equation}
    \langle \phi_i, \phi_j \rangle=\frac{1}{p} \cdot \mathcal{G}_k(a_i-a_j)
    =\frac{1}{p}\sum_{h=1}^{k-1}\chi_k^{-h}(a_i-a_j)G(\chi_k^h).
\end{equation}
\end{lemma}

\begin{proof}%[Proof of Lemma~\ref{lem-inner}]
Recall that $R_p^{(k)}=\{b_1, b_2, \ldots, b_{(p-1)/k}\}$ is the set of all non-zero $k$-th powers of $\mathbb{F}_p$. Note that for every $1 \leq l \leq \frac{p-1}{k}$, the equation $X^k \equiv b_l \pmod{p}$ has exactly $k$ distinct non-zero solutions; see e.g. \cite[p.11 and Lemma~10.4.1]{BEW1998}. 
Then we have
\begin{align}
\begin{split}
    \langle \phi_i, \phi_j \rangle
    &=\frac{1}{p}+\frac{k}{p}\sum_{l=1}^{\frac{p-1}{k}}\psi\bigl((a_i-a_j)b_l \bigr)\\
    &=\frac{1}{p}+\frac{k}{p}\cdot \frac{1}{k}\sum_{x \in \mathbb{F}_p^*}\psi\bigl((a_i-a_j)x^k \bigr)\\
    &=\frac{1}{p}\sum_{x \in \mathbb{F}_p}\psi\bigl((a_i-a_j)x^k \bigr)
    =\frac{1}{p} \cdot \mathcal{G}_k(a_i-a_j),\\
\end{split}
\end{align}
which, together with Lemma~\ref{lem-Gauss}, proves the lemma. %{\color{red}not lemma 6?}, proves the lemma.
\end{proof}

\subsection{Proof of Theorem~\ref{thm-main1}}%s of Theorem~\ref{thm-main1} and Corollary~\ref{cor-bottleneck}}

\begin{proof}[Proof of Theorem~\ref{thm-main1}]
Suppose that Conjecture~\ref{conj-pgc} holds.
Then for each $0<\alpha<1/2$ and each prime $p>p(\alpha)$, there exists some $\beta=\beta(\alpha)>0$ such that the property $\mathcal{P}(\alpha, \beta)$ holds for any non-trivial multiplicative character of $\mathbb{F}_p$.

Now fix $0<\alpha<1/2$ arbitrarily.
If $\alpha+2\beta<1/2$ holds, we may take $\beta_0=\beta$.
If $\alpha+2\beta \geq 1/2$, choose $\beta_0<\beta$ so that $\alpha+2\beta_0<1/2$; note that the property $\mathcal{P}(\alpha, \beta)$ implies the weaker property $\mathcal{P}(\alpha, \beta_0)$ since $\beta_0<\beta$.
Now take real numbers $\varepsilon_1$ and $\varepsilon_2$ with $0\leq \varepsilon_1<\varepsilon_2<\beta_0$.
Let $p>p(\alpha)$ be a prime such that there exists a factor $k$ of $p-1$ with $p^{\varepsilon_1}< k \leq p^{\varepsilon_2}$, which is possible in the light of Theorem~\ref{thm-F2008}.
Then according to Remark~\ref{rmk-main1}, we can pick a real number $\tau$ such that $\max\{\alpha+\beta_0, (1-\varepsilon_1)/2-\beta_0\}<\tau<1/2-\varepsilon_2$. 

It suffices to prove that the matrix $\Phi_p^{(k)}$ satisfies the $(p^{\tau+\beta_0}, (k-1)p^{\tau-1/2})$-flat RIP. 
Since if this is true, Proposition~\ref{prop-flatRIP} shows that $\Phi_p^{(k)}$ has the $(K, \delta)$-RIP with $K=p^{\tau+\beta_0}$ and $\delta=150 \cdot (k-1)p^{\tau-1/2} \cdot \log\big(p^{\tau+\beta_0}\big)$, implying that
$\delta=O(p^{\tau+\varepsilon_2-1/2+o(1)})$ since  $k\leq p^{\varepsilon_2}$.
This gives the desired conclusion in Theorem~\ref{thm-main1}.

To that end, recall that $\phi_i$ is the $i$-th column of $\Phi_p^{(k)}$ and $\chi_k$ is a non-trivial multiplicative character of $\mathbb{F}_p$ of order $k$. 
For every pair of disjoint subsets $I, J \subset \{1, \ldots, p\}$, we have
\begin{align}
\begin{split}
\label{eq-main1}
   \Bigl |\Bigl\langle \sum_{i \in I}\phi_i, \sum_{j\in J}\phi_j \Bigr\rangle \Bigr|
   &=\frac{1}{p} \cdot \biggl|\sum_{i \in I, j\in J}\sum_{h=1}^{k-1}\chi_k^{-h}(a_i-a_j)G(\chi_k^h) \biggr|\\
   &=\frac{1}{p} \cdot \biggl|\sum_{h=1}^{k-1}G(\chi_k^h)\sum_{i \in I, j\in J}\chi_k^{-h}(a_i-a_j) \biggr|\\
   &\leq \frac{1}{p} \cdot \sum_{h=1}^{k-1}|G(\chi_k^h)|\cdot \biggl|\sum_{i \in I, j\in J}\chi_k^{-h}(a_i-a_j) \biggr|\\
   &= \frac{1}{\sqrt{p}} \sum_{h=1}^{k-1} \biggl|\sum_{i \in I, j\in J}\chi_k^{-h}(a_i-a_j) \biggr|, 
\end{split}
\end{align}
where the first equality follows from Lemma~\ref{lem-inner}; the inequality follows from the triangle inequality; and the last equality follows from Lemma~\ref{lem-Gauss2}.  
Recall that $\chi_k^{-h}$ is a non-trivial multiplicative character of $\mathbb{F}_p$ for every $1 \leq h \leq k-1$. 
Since the condition 2) implies that $\alpha+\beta_0<1/2$,
if $|I|, |J| \leq p^{\tau+\beta_0}$, then (\ref{eq-main1}) together with Lemma~\ref{lem-doubsum} yields
\begin{align}
\label{eq-main2}
\begin{split}
    \Bigl |\Bigl\langle \sum_{i \in I}\phi_i, \sum_{j\in J}\phi_j \Bigr\rangle \Bigr|
   &\leq \frac{1}{\sqrt{p}} \cdot (k-1) \cdot  p^{\tau}\sqrt{|I||J|}\\
   &=(k-1)p^{\tau-\frac{1}{2}}\sqrt{|I||J|}.
\end{split}
\end{align}
Therefore $\Phi_p^{(k)}$ has the $(p^{\tau+\beta_0}, (k-1)p^{\tau-1/2})$-flat RIP. 
This completes the proof. 
\end{proof}

%Finally, we can obtain the following theorem which follows from a slight modification of the proof of Theorem~\ref{thm-main1}.

%\begin{theorem}
%\label{thm-bottleneck2}
%Suppose that Conjecture~\ref{conj-pgc} holds.
%Let $k\geq 2$ be a fixed integer and $p$ a sufficiently large prime such that $k|(p-1)$. 
%Then the $(p+k-1)/k \times p$ matrix $\Phi_p^{(k)}$ has the $(K, o(1))$-RIP with $K=\Omega(p^{\gamma})=\Omega(M^{\gamma})$ for some $\gamma>1/2$.
%\end{theorem}
%Note that for any fixed $k \geq 2$, by the Dirichlet's prime number theorem (see e.g. \cite{A1976}), there exist infinitely many primes $p$ such that $k|(p-1)$.

%%%%%%%%%%%%%%%%%%%%

\section*{Acknowledgement}
The authors are grateful to Professor Igor Shparlinski for his constructive comments and pointing out the references~\cite{F2008}, \cite{H2017}, \cite{S2013}, \cite{V2018}, \cite{VS2017}. 
%S. Satake has been supported by Grant-in-Aid for JSPS Fellows 20J00469 of the Japan Society for the Promotion of Science.


\begin{thebibliography}{}
%\bibitem{A1976} 
%T. M. Apostol,
%Introduction to Analytic Number Theory, 
%Springer, 1976.

\bibitem{AY2019}
A. Arian and \"{O}. Yilmaz,
``RIP constants for deterministic compressed sensing matrices-beyond
Gershgorin", 
{\it arXiv:1911.07428}, 2019.


\bibitem{BDMS2013}
A. S. Bandeira, E. Dobriban, D. G. Mixon and W. F. Sawin,
``Certifying the restricted isometry property is hard",
{\it IEEE Trans. Inf. Theory}, vol. 59, no. 6, pp. 3448--3450, June 2013.

\bibitem{BFMW2013} 
A. S. Bandeira, M. Fickus, D. G. Mixon and P. Wong,
``The road to deterministic matrices with the restricted isometry property",
{\it J. Fourier Anal. Appl.}, vol. 19, no. 6, pp. 1123--1149, 2013.

\bibitem{BMM2017} A. S. Bandeira, D. G. Mixon and J. Moreira, 
``A conditional construction of restricted isometries",
{\it Int. Math. Res. Not. IMRN}, no. 2, pp. 372--381, 2017.

\bibitem{BEW1998} B. C. Berndt, R. J. Evans and K. S. Williams, 
Gauss and Jacobi Sums, 
John Wiley \& Sons, Inc., 1998. 


\bibitem{BDFKK2011}
J. Bourgain, S. Dilworth, K. Ford, S. Konyagin and D. Kutzarova,
``Explicit constructions of RIP matrices and related problems",
{\it Duke Math. J.}, vol. 159, no. 1, pp. 145--185, 2011.

\bibitem{Ca2008}
E. Cand\`{e}s, 
``The restricted isometry property and its implications for compressed sensing",
{\it C. R. Acad. Sci. Paris, Ser. I}, vol. 346, no. 9-10, pp. 589--592, 2008.

\bibitem{C2008}
M.-C. Chang, 
``On a question of Davenport and Lewis and new character sum bounds in finite fields",
{\it Duke Math. J.}, vol. 145, no. 3, pp. 409--442, 2008.

\bibitem{CZ2016}
E. Chattopadhyay and D. Zuckerman, 
``New extractors for interleaved sources",
{\it Proceeding of 31st Conference on
Computational Complexity (CCC 2016)}, 2016.


%\bibitem{CC2017}
%F. J. Chen and Y. G. Chen, 
%``On the largest prime factor of shifted primes",
%{\it Acta Math. Sin. (Engl. Ser.)}, vol. 33, no. 3, pp. 377--382, 2017.

\bibitem{C1994}
F. R. K. Chung, 
``Several generalizations of Weil’s sums",
{\it J. Number Theory}, vol. 49, no. 1, pp. 95--106, 1994.

\bibitem{EM1964}
P. Erd\H{o}s and L. Moser,
``On the representation of directed graphs as unions of orderings",
{\it Magyar Tud. Akad. Mat. Kutat\'{o} Int. K\"{o}zl.}, vol. 9, 
pp. 125--132, 1964.

\bibitem{F2008}
K. Ford, 
``The distribution of integers with a divisor in a given interval",
{\it Ann. of Math. (2)}, vol. 168, no. 2, pp. 367--433, 2008.

%\bibitem{F1989}
%J. B. Friedlander, 
%``Shifted primes without large prime factors",
%{\it Number Theory and Applications}, pp. 393--401, 
%Kluwer Acad. Publ., 1989.

\bibitem{GM2020}
A. M. G\"{u}lo\u{g}lu and M. R. Murty, 
``The Paley graph conjecture and Diophantine $m$-tuples",
{\it J. Combin. Theory Ser. A}, vol. 170, 105155, 2020.

\bibitem{H2017}
B. Hanson, 
``Estimates for character sums with various convolutions", 
{\it Acta Arith.}, vol. 179, no. 2, pp. 133--146, 2017.

\bibitem{K2008}
A. A. Karatsuba,
``Arithmetic problems in the theory of Dirichlet characters",
{\it Russian Math. Surveys}, vol. 63, no. 4, pp. 641--690, 2008.

\bibitem{LN1994}
R. Lidl and H. Niederreiter,
Introduction to Finite Fields and Their Applications, Cambridge University Press, 1994.

\bibitem{LLKR2015}
W. Lu, W. Li, K. Kpalma and J. Ronsin,
``Compressed sensing performance of random Bernoulli matrices with high compression ratio",
{\it IEEE Signal Process. Lett.}, vol. 22, no. 8, pp. 1074--1078, Aug. 2015.


\bibitem{M2013} 
D. G. Mixon, 
``Deterministic RIP matrices: Breaking the square-root bottleneck", Short, Fat Matrices (weblog). %\url{https://dustingmixon.wordpress.com/2013/12/02/deterministic-rip-matrices-breaking-the-square-root-bottleneck/}.

\bibitem{M2015} 
D. G. Mixon, 
``Explicit matrices with the restricted isometry property: breaking the square-root bottleneck", 
{\it Compressed Sensing and Its Applications}, pp. 389–417,
Birkh\"{a}user/Springer, 2015.


\bibitem{RMKADD2020}
D. Ramalho, K. Melo, M. Khosravy, F. Asharif, M. S. S. Danish and C. A. Duque, 
``A review of deterministic sensing matrices", {\it Compressive Sensing in Healthcare}, pp. 89--110, Academic Press, 2020.


\bibitem{S2020}
S. Satake,
``On the restricted isometry property of the Paley matrix",
submitted.
%\url{https://sites.google.com/site/satakeshohei/papers}, 2020.

\bibitem{S2014}
I. D. Shkredov, 
``Sumsets in quadratic residues",
{\it Acta Arith.}, vol. 164, no. 3, pp. 221--243, 2014.

\bibitem{S2013}
I. E. Shparlinski, 
Additive decompositions of subgroups of finite fields,
{\it SIAM J. Discrete Math.}, vol. 27, no. 4, pp. 1870--1879, 2013.

\bibitem{V2018}
A. S. Volostnov, 
``On double sums with multiplicative characters",
{\it Math. Notes}, vol. 104, no. 1-2, pp. 197--203, 2018.

\bibitem{VS2017}
A. S. Volostnov, I. D. Shkredov,  
``Sums of multiplicative characters with additive convolutions",
{\it Proc. Steklov Inst. Math.}, vol. 296, no. 1, pp. 256--269, 2017. 

\bibitem{W1974} L. R. Welch,
``Lower bounds on the maximum cross correlation of signals",
{\it IEEE Trans. Inf. Theory}, vol. 20, no. 3, pp. 397--399, May 1974.

\bibitem{Z1991}  D. Zuckerman,
Computing Efficiently Using General Weak Random Sources, Ph.D. dissertation, University of California at Berkeley, US, 1991.
\end{thebibliography}
\end{document}